\documentclass[preprint,pre,twocolumn,10pt]{revtex4}%
\usepackage{amsfonts}
\usepackage{amsmath}
\usepackage{amssymb}
\usepackage{graphicx}%
\providecommand{\U}[1]{\protect\rule{.1in}{.1in}}
\newtheorem{theorem}{Theorem}

\newenvironment{proof}[1][Proof]{\noindent\textbf{#1.} }{\ \rule{0.5em}{0.5em}}
\begin{document}
\preprint{UATP/1302}
\title{Some Rigorous Results Relating Nonequilibrium, Equilibrium, Calorimetrically
Measured and Residual Entropies during Cooling }
\author{P. D. Gujrati}
\email{pdg@uakron.edu}
\affiliation{Department of Physics, Department of Polymer Science, The University of Akron,
Akron, OH 44325}

\begin{abstract}
We use rigorous nonequilibrium thermodynamic arguments to establish that (i)
the nonequilibrium entropy $S(T_{0})$ of any system is bounded below by the
experimentally (calorimetrically) determined entropy $S_{\text{expt}}(T_{0})$,
(ii) $S_{\text{expt}}(T_{0})$ is bounded below by the equilibrium or
stationary state (such as the supercooled liquid) entropy $S_{\text{SCL}%
}(T_{0})$ and consequently (iii) $S(T_{0})$ cannot drop below $S_{\text{SCL}%
}(T_{0})$. It then follows that the residual entropy $S_{\text{R}}$ is bounded
below by the extrapolated $S_{\text{expt}}(0)>S_{\text{SCL}}(0)$ at absolute
zero. These results are very general and applicable to all nonequilibrium
systems regardless of how far they are from their stationary states.

\end{abstract}
\date{\today}
\maketitle

\section{Introduction}

Nonequilibrium states like glasses from supercooled liquids (SCLs)\ are
abundant in Nature, whose entropy $S$ can only be estimated by
calorimetrically measured entropy $S_{\text{expt}}$, which can then be
extrapolated to absolute zero. The extrapolated value $S_{\text{R}}$ at
absolute zero is commonly known as the \emph{residual entropy} and is normally
found to satisfy $S_{\text{R}}>0$. In practice, one considers the isobaric
entropy $S(T_{0})$ of the system as a function of the temperature $T_{0}$ of
the surrounding medium; see Fig. \ref{Fig.System}. The existence of
$S_{\text{R}}$ was first theoretically demonstrated by Pauling and Tolman
\cite{Pauling}; see also Tolman \cite{Tolman}. In addition, the existence of
the residual entropy has been demonstrated rigorously for a very general spin
model by Chow and Wu \cite{Chow}. The residual entropy for glycerol was
observed by Gibson and Giauque \cite{Giauque-Gibson} and for ice by Giauque
and Ashley \cite{Giauque}. Pauling \cite{Pauling-ice} provided the first
numerical estimate for the residual entropy for ice, which was later improved
by Nagle \cite{Nagle}. Nagle's numerical estimate has been recently verified
by simulation \cite{Isakov,Berg}. The numerical simulation carried out by
Bowles and Speedy \cite{Speedy} for glassy dimers also supports the existence
of a residual entropy. For a brief review of the history of the residual
entropy, see
\cite{Goldstein,Gutzow-Schmelzer,Nemilov,GujratiResidualentropy,Gujrati-Symmetry,Sethna,Sethna-Paper,Johari}%
. Thus, it appears that the support in favor of the residual entropy, see the
curve Glass1 in see Fig. \ref{Fig_entropyglass}, is quite strong. Its
existence also does not violate Nernst's postulate, as the latter is
applicable only to true equilibrium states with a \emph{non-degenerate} ground
state \cite{Landau,Gujrati-Nernst,Gujrati-Fluctuations}. Indeed, many exactly
solved statistical mechanical models show a non-zero entropy at absolute zero.
However, as of yet, no experiment can be performed at absolute zero to
experimentally determine the residual entropy; in all cases, some sort of
\emph{extrapolation} is required. This point should not be forgotten in the
following whenever we speak of the residual entropy. Despite the above
mentioned support for the reality of the residual entropy, it has become a
highly debated issue in the literature
\cite{Note1,Jackel,Palmer,Hemmen,Thirumalai,Kivelson,Gupta} as discussed by
these authors. The reason for the debate is that the relationship among
$S(T_{0})$, $S_{\text{expt}}(T_{0})$ and the entropy $S_{\text{SCL}}(T_{0})$
of the corresponding stationary state is not well understood, and
understanding this relationship is the main theme of this work.
\begin{figure}
[ptb]
\begin{center}
\includegraphics[
height=1.5229in,
width=3.128in
]%
{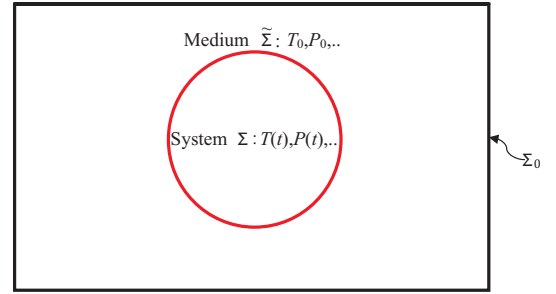}%
\caption{An isolated system $\Sigma_{0}$ consisting of the system $\Sigma$ in
a surrounding medium $\widetilde{\Sigma}$. The medium and the system are
characterized by their fields $T_{0},P_{0},...$ and $T(t),P(t),...$,
respectively, which are different when the two are out of equilibrium. }%
\label{Fig.System}%
\end{center}
\end{figure}

In the following, we will speak of "the equilibrium" state associated with a
nonequilibrium state as the stationary (time-independent) state. Depending on
the context, the equilibrium state may represent a true equilibrium state such
as a crystal or a stationary metastable state such as the supercooled liquid.
A nonequilibrium state in this work will always be taken as time-dependent.
Accordingly, $S(T_{0})$ above should be correctly expressed as  $S(T_{0},t)$,
and in some cases can be expressed as a function $S(T(t))$ of the
instantaneous temperature $T(t)$ \cite{Guj-NE-I,Guj-NE-II,Guj-NE-III} of the
system; see Fig. \ref{Fig.System}. In this work, we will not be concerned with
$T(t)$. Thus, we will simply use $S(T_{0})$ for the nonequilibrium state,
knowing well that the state continues to change with time. 

For the purpose of clarity, we will consider supercooled liquids and
associated nonequilibrium states (glasses) in the following, but the arguments
are applicable to all nonequilibrium states. The supercooled liquid undergoes
a glass transition over a transition range, see Fig. \ref{Fig_entropyglass},
over which the entropy falls rapidly with lowering temperature $T_{0}$. The
transition region is controlled by the rate of cooling so that the glass is a
nonequilibrium state \cite{Gujrati-book,Guj-NE-III}. As the irreversibility
due to the glass transition does not allow for an exact evaluation of the
entropy, it has been suggested \cite{Kivelson,Gupta} that the entropy
decreases by an amount almost equal to $S_{\text{R}}$ within the glass
transition region so that the glass (see Glass2 in Fig. \ref{Fig_entropyglass}%
, whose entropy lies below the supercooled liquid) would have a vanishing
entropy at absolute zero. It has been shown by Goldstein \cite{Goldstein} that
Glass2 results in a violation of the second law. It should be stressed that if
there is ever any \emph{conflict} between the second law \cite{note1} and any
other law in physics such as the zeroth or the third law, it is the second law
that is believed to hold in \emph{all} cases. One can also argue that to
confine the glass into a \emph{unique} basin in the energy landscape requires
\emph{microscopic} information
\cite{Note-Measurement,Gujrati-book,Gujrati-Symmetry}; hence, the particular
glass \emph{cannot} be considered in a macrostate. Oppenheim \cite{Oppenheim}
has also raised somewhat of a similar objection.

We have drawn the two entropy curves (Glass1 or Glass2) in Fig.
\ref{Fig_entropyglass} that emerge out of the entropy curve for the
equilibrated supercooled liquid for a given $\tau_{\text{obs}}$ in such a way
that Glass1 has its entropy above (so that $S_{\text{R}}\geq0$)\ and Glass2
below (so that $S_{\text{R}}\equiv0$) that of the supercooled liquid. The
entropy of Glass1 (Glass2) approaches that of the equilibrated supercooled
liquid entropy from above (below) during isothermal (fixed temperature of the
medium) relaxation; see the two downward vertical arrows for Glass1. It is the
approach to equilibrium that distinguishes the two glasses, Glass1 and Glass2.
Almost all experimental investigations leave open the possibility that Glass2
may materialize if the irreversibility is too large. Our work clarifies the
situation.
\begin{figure}
[ptb]
\begin{center}
\includegraphics[
trim=1.250626in 3.954382in 1.958635in 3.864993in,
height=2.06in,
width=3.5829in
]%
{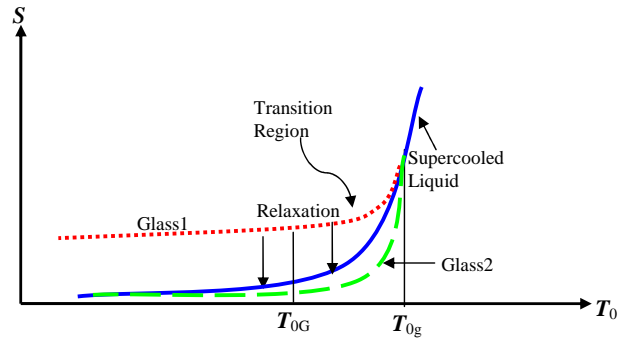}%
\caption{Schematic behavior of the entropy of the equilibrated, i.e.
stationary supercooled liquid (solid curve) and two possible
glasses\ (Glass1-dotted curve, Glass2-dashed curve) during vitrification. The
transition region between $T_{0\text{g}}$ and $T_{0\text{G}}$ has been
exaggerated to highlight the point that the glass transition is not a sharp
point.\ For all temperatures $T_{0}<T_{0\text{g}}$, any nonequilibrium state
undergoes isothermal structural relaxation in time towards the supercooled
liquid. The entropy of the supercooled liquid is shown to extrapolate to zero
per our assumption, but that of Glass1 to a non-zero value and of Glass2 to
zero at absolute zero. }%
\label{Fig_entropyglass}%
\end{center}
\end{figure}

It is abundantly clear from the above discussion that there is a need to look
at the relationship between various entropies in Fig. \ref{Fig_entropyglass}.
As is customary, we treat the supercooled liquid as an equilibrium state, even
though it not a true equilibrium state; see above. We proceed by following the
strict second law inequality $d_{\text{i}}S>0$
\cite{deGroot,Guj-NE-I,Guj-NE-II,Guj-Entropy}, see Eq.
(\ref{Second_Law_Inequality}), and use it to prove the following results
applicable to \emph{all} nonequilibrium systems, regardless of how close or
far they are from their equilibrium state:

\begin{enumerate}
\item Various entropies obey the following strict inequalities
\begin{equation}
S(T_{0})>S_{\text{expt}}(T_{0})>S_{\text{SCL}}(T_{0})\ \ \ \text{for }%
T_{0}<T_{0\text{g}},\label{GlassEntropy_Bound}%
\end{equation}
so that the entropy variation in time has a unique direction as shown by the
\emph{downward arrows} in Fig. \ref{Fig_entropyglass}. Thus, $S(T_{0})$ cannot
drop below $S_{\text{SCL}}(T_{0})$ (such as Glass2 in Fig.
\ref{Fig_entropyglass}) without violating of the second law \cite{note1}. 

\item The experimentally observed non-zero entropy at absolute zero in a
vitrification process is a \emph{strict lower bound of the residual entropy}
of any system:%
\begin{equation}
S_{\text{R}}\equiv S(0)>S_{\text{expt}}(0)>S_{\text{SCL}}%
(0).\label{ResidualEntropy_Bound}%
\end{equation}

\end{enumerate}

The Eq. (\ref{GlassEntropy_Bound}) is consistent with Glass1 but not with
Glass2. All experiments on or exact/approximate computations for
nonequilibrium systems \emph{must} obey the strict inequalities\ in Eqs.
(\ref{ResidualEntropy_Bound}-\ref{GlassEntropy_Bound}) without exception. This
is the meaning behind the usage of "... rigorous ..." in the title. The actual
values of the entropy are not relevant for the aim of this work, which is to
find the relationship among different entropies under vitrification. Because
of the possibility that the systems may be far away from equilibrium such as
in a fast quench, where the irreversible contributions may not be neglected,
our results go beyond the previous calorimetric evidence
\cite{Jackel,Gutzow-Schmelzer,Nemilov,Johari}. The systems we are interested
in include glasses and imperfect crystals as special cases. However, to be
specific, we will only consider glasses below.

\section{Entropy Bounds during Vitrification}

The vitrification process we consider is carried out at some cooling rate as
follows. The temperature of the medium is isobarically changed by some small
but fixed $\Delta T_{0}<0$ from the current value to the new value, and we
wait for (not necessarily fixed) time $\tau_{\text{obs}}$ at the new
temperature to make an instantaneous measurement on the system before changing
the temperature again. At some temperature $T_{0\text{g}}$, see Fig.
\ref{Fig_entropyglass}, the relaxation time $\tau_{\text{relax}}$, which
continuously increases as the temperature is lowered, becomes equal to
$\tau_{\text{obs}}$. Just below $T_{0\text{g}}$, the structures are not yet
frozen; they "freeze" at a lower temperature $T_{0\text{G}}$ (not too far from
$T_{0\text{g}})$ to form an amorphous solid with a viscosity close to
$10^{13}$ poise. This solid is identified as a \emph{glass}. The location of
both temperatures depends on the rate of cooling, i.e. on $\tau_{\text{obs}}$.
Over the glass transition region between $T_{0\text{G}}$ and $T_{0\text{g}}%
$\ in Fig. \ref{Fig_entropyglass}, the system gradually turns from an
equilibrium supercooled liquid at or above $T_{0\text{g}}$ into a glass at or
below $T_{0\text{G}}$ \cite{Gujrati-book,Guj-NE-III,Nemilov-Book}. We overlook
the possibility of the supercooled liquid ending in a spinodal
\cite{Gujrati-spinodal}. It is commonly believed that $S_{\text{SCL}}(0)$ will
vanish at absolute zero ($S_{\text{SCL}}(0)\equiv0$), as shown in the figure.
However, it should be emphasized that the actual value of $S_{\text{SCL}}(0)$
has no relevance for the theorems below.

We will only consider isobaric cooling (we will not explicitly exhibit the
pressure in this section), which is the most important situation for glasses.
The process is carried out along some path from an initial state A at
temperature $T_{0}$ in the supercooled liquid state which is still higher than
$T_{0\text{g}}$ to the state A$_{0}$ at absolute zero. The state A$_{0}$
depends on the path A$\rightarrow$A$_{0}$, which is implicit in the following.
The change $dS$ between two neighboring points along such a path is
$dS=d_{\text{e}}S+d_{\text{i}}S$ in modern notation
\cite{Donder,deGroot,Prigogine,Guj-NE-I,Guj-NE-II,Guj-Entropy}. The component
\begin{equation}
d_{\text{e}}S(t)=d_{\text{e}}Q(t)/T_{0}\equiv C_{P}dT_{0}/T_{0}%
\label{Heat_Capacity_Relation}%
\end{equation}
represents the reversible entropy exchange with the medium in terms of the
exchange heat $d_{\text{e}}Q(t)$ (in keeping with the modern notation) added
to the glass by the medium at time $t$ to the medium at $T_{0}$ and the heat
capacity $C_{P}$. It also represents the \emph{calorimetrically} determined
change in the entropy in any process. The component
\begin{equation}
d_{\text{i}}S>0\label{Second_Law_Inequality}%
\end{equation}
represents the irreversible entropy generation within the system in the
irreversible process, and contains, in addition to the contribution from the
irreversible heat transfer with the medium, contributions from all sorts of
viscous dissipation going on \emph{within} the system and normally require the
use of internal variables \cite{Donder,deGroot,Prigogine,Guj-NE-I,Guj-NE-II}.
The equality in Eq. (\ref{Second_Law_Inequality}) holds for a reversible
process, which we will no longer consider unless stated otherwise. A
discontinuous change in the entropy is ruled out from the continuity of the
Gibbs free energy $G$ and the enthalpy $H$ in vitrification proved elsewhere
\cite{Guj-NE-I}. Thus, we only consider a continuous change in the entropy as
shown by the two glass curves in Fig. \ref{Fig_entropyglass}.

\begin{theorem}
\label{Theorem_Lower_Bound}The experimentally observed (extrapolated) non-zero
entropy at absolute zero in a vitrification process is a \emph{strict lower
bound of the residual entropy} of any system:%
\[
S_{\text{R}}\equiv S(0)>S_{\text{expt}}(0).
\]

\end{theorem}

\begin{proof}
We have along A$\rightarrow$A$_{0}$%
\begin{equation}
S(0)=S(T_{0})+%
{\textstyle\int\limits_{\text{A}}^{\text{A}_{0}}}
d_{\text{e}}S+%
{\textstyle\int\limits_{\text{A}}^{\text{A}_{0}}}
d_{\text{i}}S,\label{General_Entropy_Calculation}%
\end{equation}
where we have assumed that there is no latent heat in the vitrification
process. Since the second integral is always \emph{positive}, and since the
residual entropy $S_{\text{R}}$ is, by definition, the entropy $S(0)$ at
absolute zero, we obtain the important result%
\begin{equation}
S_{\text{R}}\equiv S(0)>S_{\text{expt}}(0)\equiv S(T_{0})+%
{\textstyle\int\limits_{T_{0}}^{0}}
C_{P}dT_{0}/T_{0}.\label{Residual_Entropy_determination}%
\end{equation}
This proves Theorem \ref{Theorem_Lower_Bound}. The integral represents the
calorimetric contribution. 
\end{proof}

The strict forward inequality above clearly establishes that the residual
entropy at absolute zero must be strictly larger than $S_{\text{expt}}(0)$ in
any nonequilibrium process.

\begin{theorem}
\label{Theorem_Lower_Bound_SCL}The calorimetrically measured (extrapolated)
entropy during processes that occur when $\tau_{\text{obs}}<\tau
_{\text{relax}}(T_{0})$ for any $T_{0}<T_{0\text{g}}$ is larger than the
supercooled liquid entropy at absolutely zero
\[
S_{\text{expt}}(0)>S_{\text{SCL}}(0).
\]

\end{theorem}

\begin{proof}
Let $\overset{\cdot}{Q_{\text{e}}}(t)\equiv d_{\text{e}}Q(t)/dt$ be the rate
of net heat loss by the system. For each temperature interval $dT_{0}<0$ below
$T_{0\text{g}}$, we have
\begin{align*}
\left\vert d_{\text{e}}Q\right\vert  & \equiv C_{P}\left\vert dT_{0}%
\right\vert =%
{\textstyle\int\limits_{0}^{\tau_{\text{obs}}}}
\left\vert \overset{\cdot}{Q}_{\text{e}}\right\vert dt<\left\vert
dQ\right\vert _{\text{eq}}(T_{0})\\
& \equiv%
{\textstyle\int\limits_{0}^{\tau_{\text{relax}}(T_{0})}}
\left\vert \overset{\cdot}{Q}\right\vert dt,\ \ \ \ \ \ T_{0}<T_{0\text{g}}%
\end{align*}
where $\left\vert dQ\right\vert _{\text{eq}}(T_{0})>0$ denotes the net heat
loss by the system to come to equilibrium, i.e. become supercooled liquid
during cooling at $T_{0}$. For $T_{0}\geq T_{0\text{g}}$, $d_{\text{e}}Q\equiv
d_{\text{e}}Q_{\text{eq}}(T_{0})\equiv C_{P\text{,eq}}dT_{0}$. Thus,%
\[%
{\textstyle\int\limits_{T_{0}}^{0}}
C_{P}dT_{0}/T_{0}>%
{\textstyle\int\limits_{T_{0}}^{0}}
C_{P\text{,eq}}dT_{0}/T_{0}.
\]
We thus conclude that
\begin{equation}
S_{\text{expt}}(0)>S_{\text{SCL}}(0).\label{Entropy_bound_at_0}%
\end{equation}

This proves Theorem \ref{Theorem_Lower_Bound_SCL}.
\end{proof}

The strict inequalities above are the result of glass being a nonequilibrium
state. We have now verified the second statement in the Introduction.

The difference $S_{\text{R}}-$ $S_{\text{expt}}(0)$ would be larger, more
irreversible the process is. The quantity $S_{\text{expt}}(0)$ can be
determined calorimetrically by performing a cooling experiment. We take
$T_{0}$ to be the melting temperature $T_{0\text{M}}$, and uniquely determine
the entropy of the supercooled liquid at $T_{0\text{M}}$ by adding the entropy
of melting to the crystal entropy $S_{\text{CR}}(T_{0\text{M}})$ at
$T_{0\text{M}}$. The latter is obtained in a unique manner by integration
along a reversible path from $T_{0}=0$ to $T_{0}=T_{0\text{M}}$:
\[
S_{\text{CR}}(T_{0\text{M}})=S_{\text{CR}}(0)+%
{\textstyle\int\limits_{0}^{T_{0\text{M}}}}
C_{P\text{,CR}}dT_{0}/T_{0},
\]
here, $S_{\text{CR}}(0)$ is the entropy of the crystal at absolute zero, which
is traditionally taken to be zero in accordance with the third law, and
$C_{P\text{,CR}}(T_{0})$ is the isobaric heat capacity of the crystal. This
then uniquely determines the entropy of the liquid to be used in the right
hand side in Eq. (\ref{Residual_Entropy_determination}). We will assume that
$S_{\text{CR}}(0)=0$. Thus, the experimental determination of $S_{\text{expt}%
}(0)$ is required to give the \emph{lower bound} to the residual entropy in
Eq. (\ref{ResidualEntropy_Bound}). Experiment evidence for a non-zero value of
$S_{\text{expt}}(0)$ is abundant as discussed by several authors
\cite{Giauque,Giauque-Gibson,Jackel,Gutzow-Schmelzer,Nemilov,Goldstein}; a
textbook \cite{Nemilov-Book} also discusses this issue. Goldstein
\cite{Goldstein} gives a value of $S_{\text{R}}\simeq15.1$ J/K mol for
\textit{o-}terphenyl from the value of its entropy at $T_{0}=2$ K. We have
given above a mathematical justification of $S_{\text{expt}}(0)>0$\ in Eq.
(\ref{Entropy_bound_at_0}). The strict inequality proves immediately that the
residual entropy \emph{cannot} vanish for glasses, which justifies the curve
Glass1 in Fig. \ref{Fig_entropyglass}.

The inequality in Eq. (\ref{Residual_Entropy_determination}) takes into
account any amount of irreversibility during vitrification; it is no longer
limited to only small contributions of the order of $2\%$ considered by
several others
\cite{Johari,Nemilov-Book,GujratiResidualentropy,Goldstein,Betsul}, which
makes our derivation very general.

By considering the state A$_{0}$ above to be a state A$_{0}$\ of the glass in
a medium at some arbitrary temperature $T_{0}^{\prime}$ below $T_{0\text{g}}$,
we can get a generalization of Eq. (\ref{Residual_Entropy_determination}):%
\begin{equation}
S(T_{0}^{\prime})>S_{\text{expt}}(T_{0}^{\prime})\equiv S(T_{0})+%
{\textstyle\int\limits_{T_{0}}^{T_{0}^{\prime}}}
C_{P}dT_{0}/T_{0}.\label{Entropy_determination}%
\end{equation}
We again wish to remind the reader that all quantities depend on the path
A$\rightarrow$A$_{0}$, which we have not exhibited. By replacing $T_{0}$\ by
the melting temperature $T_{0\text{M}}$ and $T_{0}^{\prime}$\ by $T_{0}$, and
adding the entropy $\widetilde{S}(T_{0\text{M}})$ of the medium on both sides
in the above inequality, and rearranging terms, we obtain (with $S_{\text{L}%
}(T_{0\text{M}})=S_{\text{SCL}}(T_{0\text{M}})$ for the liquid)%
\begin{equation}
S_{\text{L}}(T_{0\text{M}})+\widetilde{S}(T_{0\text{M}})<S(T_{0}%
)+\widetilde{S}(T_{0\text{M}})-%
{\textstyle\int\limits_{T_{0\text{M}}}^{T_{0}}}
C_{P}dT_{0}/T_{0},\label{Setna_Inequality}%
\end{equation}
where we have also included the equality for a reversible process. This
provides us with an independent derivation of the inequality given by Sethna
and coworker \cite{Sethna-Paper}.

It is also clear from the derivation of Eq. (\ref{Entropy_bound_at_0}) that
the inequality can be generalized to any temperature $T_{0}<T_{0\text{g}}$
with the result%
\begin{equation}
S_{\text{expt}}(T_{0})>S_{\text{SCL}}(T_{0}),\label{Entropy_bound}%
\end{equation}
with $S_{\text{expt}}(T_{0})\rightarrow S_{\text{SCL}}(T_{0})$ as
$T_{0}\rightarrow T_{0\text{g}}$ from below. Thus, $S_{\text{expt}}(T_{0})$
appears to have a form similar to that of Glass1 in Fig.
\ref{Fig_entropyglass} but strictly lying below it. We have now verified the
first statement in the Introduction.

While we have only demonstrated the forward inequalities, the excess
$S_{\text{R}}-S_{\text{expt}}(0)$ can be computed in nonequilibrium
thermodynamics \cite{Donder,deGroot,Prigogine,Guj-NE-I,Guj-NE-II}, which
provides a clear prescription for calculating the irreversible entropy
generation. We do not do this here as we are only interested in general
results, while the calculation of irreversible entropy generation will, of
course, be system-dependent and will require detailed information. Gutzow and
Schmelzer\cite{Gutzow-Schmelzer} provide such a procedure with a single
internal variable but under the assumption of equal temperature and pressure
for the glass and the medium. However, while they comment that $d_{\text{i}%
}S\geq0$ whose evaluation requires system-dependent properties, their main
interest is to only show that it is negligible compared to $d_{\text{e}}S$.

We have proved Theorems \ref{Theorem_Lower_Bound} and
\ref{Theorem_Lower_Bound_SCL}\ by considering only the system without paying
any attention to the medium. For Theorem \ref{Theorem_Lower_Bound}, we require
the second law, i.e. Eq. (\ref{Second_Law_Inequality}). This is also true of
Eq. (\ref{Entropy_determination}). The proof of Theorem
\ref{Theorem_Lower_Bound_SCL} requires the constraint $\tau_{\text{obs}}%
<\tau_{\text{relax}}(T_{0})$ for any $T_{0}<T_{0\text{g}}$, which leads to a
nonequilibrium state. The same is also true of Eq. (\ref{Entropy_bound}).

\section{Conclusions}

We have considered the role of irreversible entropy generation during isobaric
vitrification to rigorously justify the two statements in the Introduction.
They are valid regardless of how far the system is out of equilibrium. Thus,
our results are very general and are not restricted by the small amount of
irreversibility that is normally considered in the literature. The first
statement shows that the instantaneous entropy $S(T_{0},t)$ must always be
higher than $S_{\text{expt}}(T_{0})$, which in turn must always be higher than
$S_{\text{SCL}}(T_{0})$ of the equilibrated supercooled entropy. The second
statement shows that the extrapolation of the calorimetrically measured
entropy to absolute zero forms a strict lower bound to the residual entropy
$S_{\text{R}}$. As the former is usually positive, this proves that the
residual entropy has to be at least as large as this value. From the first
statement, it also follows that Glass2 is not realistic.

The statements follow from considering the thermodynamic entropy that appears
in the second law, and their validity is not affected by which equivalent
statistical definition of entropy one may wish to use for the thermodynamic
entropy, an issue that has been investigated by us recently \cite{Guj-Entropy}%
.


\begin{thebibliography}{99}
\bibitem {note1}By the second law we mean the law of increase of the entropy
according to which the (thermodynamic) entropy of an isolated system can never
decrease. Any statistical interpretation of this law must ensure this law.
\bibitem {Note-Measurement}The preparation to get the system in a unique
microstate requires knowing precisely the cell $C_{i}$ a particular particle
$i$ belongs to. This should be contrasted with the situation such as in a
crystal or a glass, in which, although each particle is confined within a
cell, we have no knowledge which particle belongs to this cell. This ambiguity
in the microstates is a characteristic of a statistical
system\cite{GujratiResidualentropy,Gujrati-Symmetry} that results in a
macrostate. The preparation of a unique microstate results in a unique
specification and is not physically feasible to accomplish in reality.
\end{thebibliography}
\end{document}